\documentclass[a4paper,12pt]{article}
\usepackage{fullpage}
\usepackage{times}  
\usepackage{helvet}  
\usepackage{courier}  
\usepackage[hyphens]{url}  
\usepackage{graphicx} 
\usepackage{natbib}  
\usepackage{caption} 
\usepackage{algorithm}
\usepackage{algorithmic}
\usepackage[algo2e]{algorithm2e} 
\usepackage{newfloat}
\usepackage{listings}
\usepackage[margin=1in]{geometry}
\DeclareCaptionStyle{ruled}{labelfont=normalfont,labelsep=colon,strut=off} 
\floatstyle{ruled}
\newfloat{listing}{tb}{lst}{}
\floatname{listing}{Listing}

\title{GRASMOS: Graph Signage Model Selection for Gene Regulatory Networks}
\author {
    Angelina Brilliantova,
    Hannah Miller, 
    Ivona Bezáková
}

\usepackage{booktabs}       
\usepackage{amsfonts}       
\usepackage{nicefrac}       
\usepackage{microtype}      
\usepackage{xcolor}         
\usepackage{amsmath}
\usepackage{amsthm}
\usepackage{todonotes}
\usepackage{cleveref}
\usepackage{nccmath}

\newcommand{\self}[2]{\mathrm{self}_{#1}({#2})}

\newcommand{\out}[1]{\mathrm{out}(#1)}
\newcommand{\cout}[2]{\mathrm{out}_{#1}({#2})}

\newcommand{\hcoutc}[3]{\overline{\mathrm{out}}_{#1,#2}({#3})}
\newcommand{\hcinc}[3]{\overline{\mathrm{in}}_{#1,#2}({#3})}

\newcommand{\cin}[2]{\mathrm{in}_{#1}({#2})}
\newcommand{\q}[1]{q_{#1}}
\newcommand{\C}[1]{\mathcal{C}_{#1}}
\newcommand{\E}[0]{\xi}

\newcommand{\ie}{i.e.}
\newcommand{\etc}{etc.}
\newcommand{\eg}{e.g.}
\newtheorem{theorem}{Theorem}

\newtheorem{corr}[theorem]{Corollary}

\newcommand{\A}[1]{\mathcal{A}^{#1}}
\newcommand{\Lh}[1]{\mathcal{L}(#1)}
\newcommand{\Oh}[1]{\mathcal{O}(#1)}
\newcommand{\pr}[1]{\mathcal{P}^{#1}}

\date{%
    Rochester Institute of Technology
}
\begin{document}

\maketitle
\begin{abstract}
    
Signed networks, i.e., networks with positive and negative edges, commonly arise in various domains from social media to epidemiology. 
Modeling signed networks has many practical applications, including the creation of synthetic data sets for experiments where obtaining real data is difficult. Influential prior works proposed and studied various graph topology models, as well as the problem of selecting the most fitting model for different application domains. However, these topology models are typically unsigned.

In this work, we pose a novel Maximum-Likelihood-based optimization problem for modeling signed networks given their topology and showcase it in the context of gene regulation. Regulatory interactions of genes play a key role in organism development, and when broken can lead to serious organism abnormalities and diseases. Our contributions are threefold: First, we design a new class of signage models for a given topology. Based on the parameter setting, we discuss its biological interpretations for gene regulatory networks (GRNs). Second, we design algorithms computing the Maximum Likelihood -- depending on the parameter setting, our algorithms range from closed-form expressions to MCMC sampling. Third, we evaluated the results of our algorithms on synthetic datasets and real-world large GRNs. Our work can lead to the prediction of unknown gene regulations, the generation of biological hypotheses, and realistic GRN benchmark datasets.
\end{abstract}

\section{Introduction}
Networks with positive and negative edges (\textit{signed networks}) are ubiquitous across various domains. They work well for situations where objects modeled as network vertices have positive and negative interactions. Accounting for edge types helps substantially with many important network-related problems, such as missing link prediction \citep{li2017rethinking, li2020learning}, node ranking \citep{li2019supervised}, network synchronization \citep{monteiro2022fast}. Signed networks were successfully applied to model social interactions (trust/distrust in Epinions \citep{xu2019link}, epidemic spreading \citep{li2021dynamics}, political interactions between US Congressmen \citep{thomas2006get}, and gene regulation \citep{mason2009signed}.

Many frameworks for solving problems on signed networks exist -- but, to be successful in real-world applications, their performance needs to be benchmarked on realistic synthetic datasets, which remain extremely scarce for signed networks. Moreover, the difference between network structures from different domains might affect the assessment of algorithms' applicability and performance. Therefore, there is a need for generative domain-specific models to create new observations from the probability distribution underlying real-world instances and help algorithms avoid overfitting.

One of the fields affected by the scarcity of realistic signed network datasets is gene regulation. The diversity of cells, organs, and eventually organisms arises from simple regulatory interactions between genes through their transcripts.
If a transcript of a gene increases the production of the transcript of a target gene, the corresponding edge is modeled as positive; if the production decreases, the edge is modeled as negative. The set of all regulatory interactions along with a set of genes is called a \textit{gene regulatory network} (GRN). GRNs play a key role in the process of organism development, and when broken can lead to serious organism abnormalities and diseases \citep{alon2019introduction}. Understanding gene regulation mechanisms as well as identifying drug interventions often rely on reconstructing GRNs from the dynamics of the corresponding gene transcripts \citep{zhang2017similarity, lopes2020regulatory}. 

Plenty of GRN reconstruction algorithms were designed \citep{weighill2021gene, margolin2006aracne, pratapa2020benchmarking}, as well as some performance evaluation pipelines, such as GeneNetWeaver \citep{schaffter2011genenetweaver} and BEELINE \citep{beeline}. However, the synthetic datasets that such pipelines create have several unfavorable assumptions, jeopardizing any conclusion they may have derived. Some of such pipelines use subsampling from known GRNs, basically creating train datasets from the test ones, which potentially leads to severe overfitting and lack of generalization to unknown samples even from the same graph distribution \citep{schaffter2011genenetweaver}. Other common approaches include using small literature-curated models as ground truth and unsigned network topologies that lack the gene regulation signs \citep{beeline} and can not be used to assess the accuracy of sign prediction. 

There are many generative models for unsigned networks, including preferential attachment \citep{barabasi1999emergence}, Erd\"os-R\'enyi \citep{erdHos1960evolution}, Kronecker \citep{leskovec2010kronecker}, and directed scale-free \citep{bollobas2003directed}, but only few generative models for signed networks \citep{jung2020balansing, derr2018signed}. For the unsigned networks, there have been developed model selection frameworks assessing the fit of the model to the dataset based on the maximum likelihood \citep{leskovec2010kronecker, bezakova2006graph}. Such frameworks aim to find the model and its parameters that have the highest probability of generating the dataset instance and allow for a more rigorous comparison of models compared to analyzing high-level network characteristics of synthetic and real-world instances. However, to the best of our knowledge, there are no model selection results for signed networks. 

In this work we designed a new class of signage models for gene regulatory networks, motivated by underlying biological processes. We then show how to select the model and its parameters that best fit the real-world datasets, using maximum likelihood. We call this optimization problem the Graph Signage Model selection for gene regulatory networks, GRASMOS.

Our signage models work for scenarios in which the graph topology is formed first and later is refined with the signs for the edges based on the nodes’ attributes. A real-life example is forming social interactions in a closed community (e.g. dormitories) in which you get acquainted with the people you live and interact with and later decide on your attitude to them. Some evolutionary hypotheses suggest a similar origin of gene regulation, with the genome structure defining the gene interactions in the course of species divergence \citep{bylino2020evolution, wittkopp2012cis}.  
Signage models facilitate the comparison of node interactions regardless of the underlying topology -- this might be useful for graphs generated with different topologies but following the same signage model,  allowing topology-agnostic comparisons and testing the plausibility of node interactions hypotheses. We acknowledge that the signage can depend on the topology. A natural next step for this work will be comparing the likelihood of a one-step approach with our two-step approach, \ie, comparing the cases where the signage is generated jointly with the network versus independently.

\paragraph{Contributions} 
1) We design a class of graph signage models with a latent (hidden) vertex group assignment, which gets generated first, followed by the edge signage, the distribution of which is conditional on the vertex group assignment. 2) We pose a new optimization problem for modelling graph signage for gene regulatory networks, GRASMOS, which aims to select the model and its parameters that best fit the observed data using Maximum Likelihood.
3) We analyze the GRASMOS problem space and identify four cases of parameter combinations that have different intuition and induce different computational complexity of the likelihood computation. We give efficient algorithms for estimating the likelihood corresponding the given parameters. Depending on the parameter class they belong to, our algorithms range from closed-form expression to MCMC sampling. 
4) We obtain 16-fold reduction in the MCMC algorithm runtime, by identifying parts of the data sets where MCMC can be replaced by exact computation and by accounting for symmetry in the parameter space.
5) We evaluate our framework on two real-world bacteria GRN datasets - \textit{E.coli} and \textit{B.subtilis} (both with thousands of nodes and edges), as well on synthetic datasets of comparable size.

\paragraph{Related work}
Works of \citep{derr2018signed, jung2020balansing} designed generators of signed networks based on the structural balance theory, which states that certain signed triangles (balanced) are more widespread than the other triangles (unbalanced). Recent works designed models generating a node embedding in signed networks and tested its performance for missing link prediction and node classification \citep{li2017rethinking, li2020learning}. Such algorithms aim to have high predictive power while in our work we aim for a model with high explanatory power.

\section{Preliminaries}


A \emph{signed} graph $G^{\pm} = (V, E, \A{})$ consists of a directed graph $G=(V,E)$ with vertex set $V$ and edge set $E\subseteq V\times V$, and a signage function $\A{}:E\to\{+,-\}$ determining a positive or a negative sign for each edge. 
Let $n=|V|$ be the number of vertices and $m=|E|=m_++m_-$ be the number of edges, where $m_s = |\{(u,v)\in E~|~\A{}(u,v)=s\}|$ is the number of edges with sign $s$.

For any $v\in V$ let $\cout{s}{v} = |\{(v,u) \in E~|~\A{}(v,u)=s\}|$ represent the number of outgoing edges from $v$ of type $s \in \{+, -\}$ and $\cin{s}{v} = |\{(u,v) \in E~|~\A{}(u,v)=s\}|$ represent the number of incoming edges of type $s$ to vertex $v$. For each sign $s \in \{+,-\}$, we define $\bar{s}$ as its complementary sign. Therefore, if $s=+$, then $\bar{s} = -$ and vice versa. The total number of outgoing edges from $v$ is $\out{v} = \cout{s}{v} + \cout{\bar{s}}{v}$. Let $\self{+}{v}$ represent the number of $+$ self-loops of vertex $v$, $\self{-}{v}$ represent the number of $-$ self-loops.

\section{Graph Signage Model}
We propose a graph signage model, where the signs of the edges of a given directed graph $G=(V,E)$ are driven by a random latent node partition $C$ and a parameter matrix $\E$. In particular, let $S$ be a set of symbolic node group labels, $q$ be a distribution over $S$ (\ie, $\q{}:S\to[0,1]$ where $\sum_{s\in S}q(s)=1$), and $\E$ be an $|S| \times |S|$ stochastic matrix $\xi_{x,y}$, where $x,y \in S$. The signage model first randomly creates a latent node partition $C:V\to S$ by assigning each vertex independently to one of the groups in $S$ according to the distribution $\q{}$. Then, for each edge $(u,v)\in E$, the model assigns the sign $+$ to this edge with probability $\xi_{C(u),C(v)}$, and the sign $-$ otherwise, obtaining a signed graph $G^{\pm}$. The parameters of the model are combined in the tuple $\Theta = (\E, \q{})$.

For convenience, we define $|S|\times |S|$ probability matrices $\pr{+}$ and $\pr{-}$ denoting the probabilities that, based on the node assignment of the end-points of an edge, the edge gets sign $+$ or $-$. Formally, for $s_1,s_2\in S$, let $\pr{+}_{s_1,s_2}:=\xi_{s_1,s_2}$ and $\pr{-}_{s_1,s_2}:=1-\xi_{s_1,s_2}$. In this work, we focus on gene regulatory networks (GRNs), where $S=\{A,R\}$, representing activator and repressor vertices.

\section{Graph Signage Model Selection, GRASMOS}
In this section, we pose a novel optimization problem, GRASMOS, aimed at finding the parameters of the generative graph signage model with the best explanatory power of the observed data, based on maximum likelihood. Recall that the signage model only generates the signs of the edges, not the topology; in a sense, it `augments' the edges with signs based on the existing graph topology.

Formally, given a signed graph $G^{\pm}=(V,\A{})$, we are looking for such combination of parameters $\Theta$ that has the highest probability of generating edge signs $\A{}$, denoted $P(\A{}|\Theta)$. For a fixed node group assignment, $C$, we can find the corresponding probability of the signs as:

\begin{align} \label{eq:lh_one_c}
    \Lh{\Theta|C} &:= P(\A{}|\Theta, C)=
    \prod_{(u,v)\in E} \pr{\A{}(u,v)}_{C(u),C(v)}
\end{align}

In our model, we assume that the groups of vertices are unknown, so the probability of interest is a sum of conditional probabilities over all possible node assignments:

\begin{equation} \label{eq:lh_all_c}
    \mathcal{L}(\Theta)= P(\A{}|\Theta) = \sum_{C \in \C{}}P(\A{}|\Theta, C) \cdot P(C)
\end{equation}

GRASMOS aims to find model parameters $\Theta_{\mathrm{MLE}}$ 
with the highest probability (or likelihood) of realizing the edge signs $\A{}$. As is standard for the maximum likelihood approach, for numerical purposes we frame the optimization problem as finding the parameters $\Theta$ with the lowest negative log-likelihood:

\begin{align} \label{eq:optimization_problem1}
    \Theta_{\mathrm{MLE}} &= \arg \min_{\Theta}[-\log \mathcal{L}(\Theta)]  
\end{align}

\subsection{Analysis of the GRASMOS Parameter Space}
GRASMOS is an optimization problem over an $(|S|^2+|S|-1)$-dimensional parameter space, since $\E$ is of dimensions $|S|\times |S|$ and $q$ is determined by $|S|-1$ probabilities.

While our signage model is general, in this work we focus on modeling GRNs where $S=\{A,R\}$. Therefore, from now on we assume $S=\{A,R\}$ and, therefore, the optimization is over a 5-dimensional parameter space.
For each parameter setting $\Theta$, the likelihood computation is potentially an exponential summation over $2^n$ node group assignments. However, on closer scrutiny, it turns out that points in some areas of the problem space are easier to estimate that the rest, and it is the relationship between $\E$ probabilities what determines the computational difficulty of the point estimation. In this section, we subdivide our main optimization problem (\ref{eq:optimization_problem1}) into several cases (models) according to the relationship between probabilities in $\E$ and provide intuition for each model. We present the models from the simplest to the most general.

\paragraph{Node-oblivious model (NO)}
For the problem space points with identical $\E$ entries, the node group partition does not affect the likelihood of $\Theta$, so we can estimate the likelihood of such points analytically. We call this case the \textit{node-oblivious} model, because the signs of edges are independent of the nodes' groups. Define $\xi = \xi_{A,A} = \xi_{A,R} = \xi_{R,A} = \xi_{R,R}$. Then, $\pr{+}_{s_1,s_2} = \xi$ and $\pr{-}_{s_1,s_2} = 1-\xi$ for every $s_1,s_2\in S$. This model can be viewed as a signage analogue of the well-studied Erd\"os-R\'enyi random graph model.

\paragraph{Source-consistent model (SC)}
In parameter subspace, where $\E$ matrix has pairwise identical probabilities $\xi_{A,A} = \xi_{A,R} =: \xi_{A*}$ and $\xi_{R,A} = \xi_{R,R} = \xi_{R*}$. The edge sign probabilities depend only on the source node's group. In the context of gene regulation, that means that a gene mostly performs the regulation of the same sign: an activator gene tends to activate, and a repressor gene represses.  We call the parameter instances belonging to this case the \textit{source-consistent model}. This observation allows us to reformulate (\ref{eq:lh_one_c}) as a product of probabilities of signed outgoing edges over the set of vertices: 

\begin{align}\label{eq: lh_one_color_sc}
     P(\A{}|\Theta, C) = \prod_{v \in V} \xi_{C(v)*}^{\cout{+}{v}}\cdot (1-\xi_{C(v)}*)^{\cout{-}{v}}
\end{align}

\paragraph{Target-consistent model (TC)}
Similarly to the previous model, in the \textit{target-consistent} model, the edge probabilities related to the same source group are pairwise equal, i.e. $\xi_{A,A} = \xi_{R,A} =: \xi_{*A}$ and $\xi_{A,R} = \xi_{R,R} =: \xi_{*R}$, and the edge probabilities depends only on the target node's group. In the context of gene regulation, that means that a gene is either predominantly repressed or activated by other genes.. The likelihood of $\Theta$ for this model can be found as the product of probabilities of signed incoming edges:

\begin{align}\label{eq: lh_one_color_tc}
     P(\A{}|\Theta, C) = \prod_{v \in V} \xi_{*C(v)}^{\cin{+}{v}}\cdot (1-\xi_{*C(v)})^{\cin{-}{v}}
\end{align}

\paragraph{Bi-node-consistent model (BNC)}
In the \textit{bi-node consistent} model, the edge signs depend on the group assignment of both the source and the target nodes and the probabilities in $\E$ can be arbitrary. In this parameter subspace, we can see $\E$ instances inducing assortativity: \eg, whether nodes from the same group tend to have positive edges between each other, and negative edges to the nodes from the other group \citep{mussmann2015incorporating}, \citep{newman2002assortative}. We reformulate (\ref{eq:lh_one_c}) for the likelihood estimation of the BNC parameters as a product of the likelihood contributions of non-self-loop (outgoing) signed edges and self-loop signed edges for each vertex given vertex assignment $C$. Let $l(v, C)^{\mathrm{out}}$ 
and $l(v, C)^{\mathrm{self}}$ be the likelihood contribution of the outgoing 
and self-loop edges of vertex $v$ to the $\Theta$ likelihood given $C$ in (\ref{eq:lh_one_c}). Thus, 
\begin{align}
    &l(v, C)^{\mathrm{out}} = \prod_{u: u \neq v, (v,u)\in E} \pr{\A{}(v,u)}_{C(v), C(u)}  \\
    &l(v, C)^{\mathrm{self}} = \xi_{C(v),C(v)}^{\self{+}{v}} \cdot  (1-\xi_{C(v),C(v)})^{\self{-}{v}}
\end{align}
We rewrite (\ref{eq:lh_one_c}) with respect to the contribution of each vertex to the overall likelihood:



\begin{align}\label{eq: lh_one_color_bnc}
     P(\A{}|\Theta, C) = \prod_{v \in V} l(v, C)^{\mathrm{out}} 
     \cdot l(v,C)^{\mathrm{self}}
\end{align}

\paragraph{Problem space symmetry}
The GRASMOS problem space is symmetrical with respect to the likelihood values, and for any parameter $\Theta$ there is a unique parameter $\Theta'$ that has exactly the same likelihood for any given signed graph.

\begin{theorem}
For every $\Theta = (\E, \q{})$, there is $\Theta' = (\E', \q{})$ such that $\Lh{\Theta} = \Lh{\Theta'}$ for any given signed graph $G^\pm = (V, E, \A{})$ and $\Theta'$ is obtained by setting $\xi'_{RR} = \xi_{AA}$, $\xi'_{AA} = \xi_{RR}$, $\xi'_{RA} = \xi_{AR}$, $\xi'_{AR} = \xi'_{RA}$, $\q{A}'=1-\q{A}$.
\end{theorem}

\begin{proof}
Consider a $\Theta = (\xi_{AA}, \xi_{AR}, \xi_{RA}, \xi_{RR}, q)$ parameter and its counterpart $\Theta' = (\xi'_{AA}, \xi'_{AR},$ $\xi'_{RA}, \xi'_{RR}, q')$ where $\xi'_{RR} = \xi_{AA}$, $\xi'_{AA} = \xi_{RR}$, $\xi'_{RA} = \xi_{AR}$, $\xi'_{AR} = \xi'_{RA}$, $q'=1-q$. We will show that $\Lh{\A{}|\Theta} = \Lh{(\A{}|\Theta')}$. According to equation $\ref{eq:lh_all_c}$, $\Lh{\A{}|\Theta} = \sum_{C \in \C{}}P(\A{}|\Theta, C)\cdot P(C)$. Likewise, $\Lh{\A{}|\Theta'} = \sum_{' \in \C{}}P(\A{}|\Theta', C)\cdot P(C)$. For $\forall C \in \C{}$ there exist a unique $C' \in \C{}$ such that each vertex in $C$ has a complementary assignment in $C'$: $\forall v \in V, C'(v) = \bar{C}(v)$. Consider a pair of such vertex assignments $C, C'$ and their contribution to $\Lh{\A{}|\Theta}$ and $\Lh{\A{}|\Theta'}$ correspondingly. We will compare the contribution of $C$ to $\Lh{\A{}|\Theta}$, to the contribution of $C'$ to $\Lh{\A{}|\Theta'}$ and show that they are identical. In $C$ each edge $(u,v)\in E$, contributes $\mathcal{P}_{C(u), C(v)}^{\A{}(u,v)}$ to $\Lh{\A{}|\Theta}$, while under $C'$ each edge $(u,v)\in E$, contributes $\mathcal{P}_{\bar{C}(u), \bar{C}(v)}^{\A{}(u,v)}$. But for any $X,Y \in S$, $\xi_{X,Y} = \xi_{\bar{X}, \bar{Y}}$ by the definition of $\Theta'$, therefore $\mathcal{P}_{C(u), C(v)}^{\A{}(u,v)} = \mathcal{P}_{\bar{C}(u), \bar{C}(v)}^{\A{}(u,v)}$ and the contribution of any edge is identical to both the $C$ and $C'$ terms in the likelihoods $\Lh{\Theta}$ and $\Lh{\Theta'}$, respectively. Therefore, the contribution of all edges probabilities in $C$ to $\Lh{\A{}|\Theta}$ is equal to the contribution of all edges probabilities in $C'$ to $\Lh{\A{}|\Theta'}$.  Now, it suffices to show that the contribution of vertex probabilities is identical too. Consider the contribution of $P(C)$ and $P(C')$. $P(C) = q^{V_A} \cdot (1-q)^{V_R}$, $P(C') = q^{V'_A} \cdot (1-q)^{V'_R}$, where $V_A, V'_A$ ($V_R, V'_R$) are the numbers of activators (repressors) in $C, C'$ respectively. Since each vertex in $C$ has the opposite assignment in $C'$, $V_A = n - V'_A = V'_R$, $V_R = n-V'_R = V'_A$, and $q = 1-q'$ by the premise. Therefore, $P(C) = (1-q)^{V'_R}\cdot q^{V'_A} = P(C')$. We proved that the contribution of $C$ to $\Lh{\A{}|\Theta}$ is equal to the contribution of $C'$ to $\Lh{\A{}|\Theta'}$. Any $C \in \C{}$, has its complement $C' \in \C{}$, and therefore the total likelihood of $\Theta$ as a sum of $2^n$ likelihoods of $C$ will be equal to the total likelihood of $\Theta'$ as a sum of $2^n$ likelihoods of the corresponding $C'$.
\end{proof}

This allows us to reduce the search space of the optimization problem by a half. To simplify our notation, we omit the subscript in $\q{A}$ in the remainder of this paper.

\section{Likelihood Estimation of the GRASMOS Model Parameters}

In this section we show how to estimate the likelihood of the GRASMOS model parameters, depending on which of the above-stated models they belong to. The methods we design range in their complexity depending on the computational difficulty of the problem: we provide a closed form expression for the Node-oblivious model, a polynomial-time exact algorithm for the Source- and Target-consistent models, and an MCMC sampling algorithm for the Bi-node-consistent model.
We prove that the models (or, more precisely, their natural generalizations) are \emph{self-reducible} \citep{DBLP:journals/tcs/JerrumVV86}, which allows us to estimate the likelihood for a given $\Theta$ through a product of likelihood ratios of progressively smaller instances. For the BNC model, we use sampling to get an approximation of each of the likelihood ratios.

\subsection{Node-oblivious model}
For the NO model, from (\ref{eq:lh_one_c}) and (\ref{eq:lh_all_c}) we get $P(\A{}|\Theta, C) = P(\A{}|\Theta)  = \xi^{m_+}\cdot (1-\xi)^{m_-}$.
Therefore, the optimization problem (\ref{eq:optimization_problem1}) can be solved analytically: it is maximized for $\xi_{\max} := \frac{m_+}{m_++m_-}$.
Since $\q{}$ does not contribute to $P(\A{}|\Theta)$, $\Theta_{\mathrm{MLE}}=(\xi_{\mathrm{MLE}},\q{\mathrm{MLE}})$, where $\xi_{\mathrm{MLE}}$ has all entries equal to $\xi_{\max}$ and $\q{\mathrm{MLE}}$ is arbitrary.

\subsection{Likelihood estimation through the product of likelihood ratios}

The total likelihood of $\Theta$ consists of the sum of $2^n$ likelihoods of $\Theta$ conditional on the node assignment (see (\ref{eq:lh_all_c})). We express $\Lh{\Theta}$ as the product of ratios of pairs of likelihoods on decreasingly smaller spaces of node group assignments: the reduction in space is achieved by fixing the groups for some nodes.



Let $V=\{v_1,v_2,\dots,v_n\}$. For $s_1,\dots,s_j\in S$, we define $\C{j}^{[s_1,\dots,s_j]}$ as the set of node assignments $C:V\to \{A,R\}$ such that $C(v_i)=s_i$ for every $i\leq j$ (vertices $v_1,\dots,v_j$ have their group determined by $s_1,\dots,s_j$). Notice that $\C{0}$ is the set of all node assignments (with no restrictions) and that $|\C{j}^{[s_1,\dots,s_j]}|=2^{n-j}$ for any $s_1,\dots,s_j$.


For our self-reducibility approach, let us fix a ``master'' node assignment $\tilde{C}:V\to\{A,R\}$ that gradually more and more vertices will adhere to. Let
\begin{equation}\label{eq:Z_j}
    Z_j := \sum_{C_j \in \C{j}^{[\tilde{C}(v_1),...,\tilde{C}(v_j)]}} P(\A{}|\Theta,C_j) \cdot P(C_j|j),
\end{equation}
where $P(C_j|j)$ is the probability that vertices $v_{j+1},\dots,v_n$ get the node assignment given by $C_j$. Notice that $P(C_j|j)$ is a probability distribution over the node group assignment subspace $\C{j}^{[\tilde{C}(v_1),...,\tilde{C}(v_j)]}$, and, therefore, $Z_j$ is the likelihood of $\Theta$ restricted to this subspace. 

 
Our overall goal is to estimate $\Lh{\Theta}=Z_0$. We do this via the following product of likelihood ratios:
\begin{equation}\label{eq:self_red}
    \frac{Z_1}{Z_0} \cdot \frac{Z_2}{Z_1} \cdot \frac{Z_3}{Z_2} ... \cdot \frac{Z_{n}}{Z_{n-1}} = \frac{Z_n}{Z_0}.
\end{equation}
Notice that $Z_n=P(\A{}|\Theta,\tilde{C})$ can be easily computed via (\ref{eq:lh_one_c}). Therefore, if we estimate each of the ratios $\sigma_j:=\frac{Z_{j}}{Z_{j-1}}$, we can compute $Z_0$ as follows: 



\begin{equation}\label{eq:partition}
    P(\A{}|\Theta) = Z_0 = \frac{Z_n}{\prod_{j=1}^{n} \sigma_j}
\end{equation}

For numerical stability, we will choose $\tilde{C}$ so that each ratio $\sigma_j$ is reasonably far from $0$.

\begin{theorem} \label{thm:1/2}
There exists $\tilde{C}$ such that $\sigma_j\geq 1/2$ for every $j\in\{1,\dots,n\}$.
\end{theorem}

\begin{proof}
We will define $\tilde{C}$ inductively. First, notice that the ratio $\sigma_j$ does not rely on the entire $\tilde{C}$ but only on its restriction to $\{v_1,\dots,v_j\}$. For $j\in\{1,\dots,n\}$, suppose $\tilde{C}(v_1),\dots,\tilde{C}({v_{j-1})}$ have been chosen so that $\sigma_i\geq 1/2$ for every $i<j$. We will show how to choose $\tilde{C}(v_j)$ so that $\sigma_j\geq 1/2$ by considering the two possible cases: $\tilde{C}(v_j)$ can be either an activator or a repressor. For $s\in\{A,R\}$, define $Z_j^s$ using (\ref{eq:Z_j}) with $[\tilde{C}(v_1),\dots,\tilde{C}(v_{j-1}),s]$. Let $\sigma_j^s:=\frac{Z_j^s}{Z_{j-1}}$. Recall that $P(C|j-1) = \q{}^{A(C,j)}(1-\q{})^{n-j-A(C,j)}$, where $A(C,j):=|\{i~|~C(V_i)=A,~i\geq j\}|$ stands for the number of activators among $v_j,\dots,v_n$ in $C$.
Notice that
\begin{eqnarray*}
Z_{j-1} &=& \sum_{C \in \C{j-1}^{[\tilde{C}(v_1),...,\tilde{C}(v_{j-1})]}} P(\A{}|\Theta,C) \cdot P(C|j-1)\\
&=& \sum_{C \in \C{j}^{[\tilde{C}(v_1),...,\tilde{C}(v_{j-1}),A]}} P(\A{}|\Theta,C)\cdot P(C|j-1) +\\
&& \sum_{C \in \C{j}^{[\tilde{C}(v_1),...,\tilde{C}(v_{j-1}),R]}} P(\A{}|\Theta,C)\cdot P(C|j-1)\\
&=& Z_j^A\q{} + Z_j^R(1-\q{}).
\end{eqnarray*}
Therefore, $\q{}\sigma_j^A + (1-\q{})\sigma_j^R=1$. This equation cannot hold if both $\sigma_j^A<1/2$ and $\sigma_j^R<1/2$. Therefore, choosing $\tilde{C}(v_j):=\arg\max_{s\in\{A,R\}}\sigma_j^s$ ensures that $\sigma_j = \max\{\sigma_j^A, \sigma_j^R\}\geq 1/2$.
\end{proof}

The proof builds $\tilde{C}$ constructively but the algorithmic efficiency is unclear (it involves summations over exponentially many terms). In the following sections, we will show how to construct $\tilde{C}$ and compute all the $\sigma_j^*:=\max\{\sigma_j^A, \sigma_j^R\}$ (and thus $\Lh{\Theta}$) efficiently.




\subsection{Source-consistent and target-consistent models}
In the following section we show how to estimate the ratio of likelihoods $\sigma_j^{*}$ for the SC and TC models in linear time, $\Oh{n+m}$. Our algorithms rely on the fact that, in the SC and TC models, the group assignment of any vertex, $\tilde{C}(v_j)$, only affects the probabilities of edges connected to $v_j$. Therefore, the calculation of $\sigma_j^{*}$ depends only on vertex $v_j$ and its incident edges.

\begin{theorem}\label{thm:SC_alpha}
When $\Theta$ belongs to the source-consistent model, for every $j \in \{1,\dots, n\}$ and every $\tilde{C}:V\to\{A,R\}$:
$$\sigma_j^* = \max\{ \frac{1}{\q{} + \alpha_j(1-\q{})}, \frac{\alpha_j}{\q{} + \alpha_j(1-\q{})}\},$$
where
$$\alpha_j := (\frac{\xi_R*}{\xi_A*})^{\cout{+}{v_j}} \cdot (\frac{1-\xi_R*}{1-\xi_A*})^{\cout{-}{v_j}}.$$
\end{theorem}

\begin{proof}
Recall that $\sigma_j^* = \max\{\sigma_j^A, \sigma_j^R\}$, where $\sigma_j^A=\frac{Z_j^A}{Z_{j-1}}$ and $\sigma_j^R=\frac{Z_j^R}{Z_{j-1}}$.
We will first relate $Z_j^A$ and $Z_j^R$.
Consider a pair of vertex group assignments from $\C{j}^{[\tilde{C}(v_1),..., \tilde{C}(v_j)]}$, that differ only by the group of $v_j$:  $C^A \in \C{j}^{[\tilde{C}(v_1),...,\tilde{C}(v_{j-1}),A]}$ and $C^R\in \C{j}^{[\tilde{C}(v_1),...,\tilde{C}(v_{j-1}),R]}$. Both $C^A$ and $C^R$ have the first $j$ and the last $n-j-1$ vertices assigned identically and $v_j$ is an activator in $C^A$ and a repressor in $C^R$. 
Then:
\begin{align}\label{eq:alpha_SC}
      &\frac{P(\A{}|\Theta,C^R)\cdot P(C^R|j)}{P(\A{}|\Theta, C^A)\cdot P(C^A|j)} = \nonumber \\
      &= \frac{\prod_{v \in V} \xi_{C^R(v)*}^{\cout{+}{v}}\cdot (1-\xi_{C^R(v)}*)^{\cout{-}{v}}}{\prod_{v \in V} \xi_{C^A(v)*}^{\cout{+}{v}}\cdot (1-\xi_{C^A(v)}*)^{\cout{-}{v}}}\cdot \frac{P(C^R|j)}{P(C^A|j)} \nonumber \\
      &=  \frac{\xi_{R*}^{\cout{+}{v_j}}\cdot (1-\xi_{R*})^{\cout{-}{v_j}}\cdot }{\xi_{A*}^{\cout{+}{v_j}} (1-\xi_{A*})^{\cout{-}{v_j}}} \nonumber \\
      &= \alpha_j \nonumber,
\end{align}
where the first equality comes from (\ref{eq: lh_one_color_sc}) and the second equality reflects cancellation of terms for all vertices except $v_j$. 
Relating $Z_j^A$ and $Z_j^R$, we get:
\begin{eqnarray*}
    Z_j^R &=& \sum_{C^R \in \C{j}^{[\tilde{C}(v_1),...,\tilde{C}(v_{j-1}),R]}} P(\A{}|\Theta,C^R)\cdot P(C^R|j)\\
    &=& \sum_{C^A \in \C{j}^{[\tilde{C}(v_1),...,\tilde{C}(v_{j-1}),A]}} \alpha_j \cdot P(\A{}|\Theta,C^A)\cdot P(C^A|j)\\
    &=& \alpha_j \cdot Z_j^A.
\end{eqnarray*}
This allows us to compute $\sigma_j^A$ explicitly:
\begin{eqnarray*}
    \sigma_j^A &=& \frac{Z_j^A}{Z_{j-1}} = \frac{Z_j^A}{Z_j^A\q{} + Z_j^R(1-\q{})} = \frac{1}{\q{} + \alpha_j(1-\q{})}.
\end{eqnarray*}
Analogously, $\sigma_j^R = \frac{\alpha_j}{\q{} + \alpha_j(1-\q{})}$, and, therefore, $\sigma_j^* = \max\{\sigma_j^A, \sigma_j^R\}$ can be computed as stated in the theorem.
\end{proof}

\begin{corr}\label{corr:TC_Z}
For the source-consistent model, the likelihood of $\Theta$, $\Lh{\Theta}$, can be calculated in $\Oh{n+m}$ time.
\end{corr}



\begin{proof}
In time $\Oh{n+m}$ we can precompute $\cout{+}{v}$ and $\cout{-}{v}$ for every vertex $v$.
Then, for every $j$, the computation of $\alpha_j$ takes $\Oh{\out{v}}$ time and the computation of $\sigma_j^*$ takes constant time.
By (\ref{eq:partition}),
$$Z_0 = \frac{Z_n}{\prod_{j\in\{1,\dots,n\}}\max\{ \frac{1}{\q{} + \alpha_j(1-\q{})}, \frac{\alpha_j}{\q{} + \alpha_j(1-\q{})}\}},$$
where $Z_n = \prod_{v\in V}\xi_{\tilde{C}(v)*}^{\cout{+}{v}}\cdot (1-\xi_{\tilde{C}(v)*})^{\cout{-}{v}}$. 

Recall that $\tilde{C}$ is defined as $\tilde{C}(v_j):=\arg\max_{s\in\{A,R\}}\sigma_j^s$, so each $\tilde{C}(v_j)$ can be computed in constant time. The computation of $Z_n$ takes $\Oh{\sum_{v\in V}(1+\out{v})}=\Oh{n+m}$ time. Since $\Lh{\Theta}=Z_0$, the corollary follows.
\end{proof}

An analogous theorem and corollary holds for the target-consistent model (included in the Appendix).

\subsection{MCMC sampling for likelihood estimation of the bi-node-consistent model}

 In contrast to the SC and TC models, in which an edge sign depends only on one endpoint and the ratio of likelihoods $\sigma^*_j$ can then be computed analytically, we are not aware of any analytical approach for $\sigma_j^*$ in the BNC model where both edge end-points influence the sign of the edge. The main difficulty arises from the fact that the calculation of $\sigma^*_j$ under BNC might depend on the colors of the remaining $n-j-1$ nodes, which, at that point, are still unassigned in the course of the algorithm. Instead of exact computation, we estimate 
 each $\sigma^*_j$ via a Markov Chain Monte-Carlo (MCMC) sampling of vertex group assignments on the subspace corresponding to $Z_j$, with first $j$ vertices assigned to groups.
 
Suppose $\tilde{C}(v_1),\dots,\tilde{C}(v_{j})$ has been already defined. Let us define $w_j(C_j) := P(\A{}|\Theta, C_j) \cdot P(C_j|j)$
as the weight of the vertex group assignment $C_j \in \C{j}^{[\tilde{C}(v_1),...,\tilde{C}(v_j)]}$. To compute the likelihood $Z_j=\sum_{C_j \in \C{j}^{[\tilde{C}(v_1),...,\tilde{C}(v_j)]}} w_j(C_j)$, we will randomly generate $C_j$, with probability proportional to its weight. Therefore, the stationary distribution of the vertex group assignments should be $\mu_j(C_j):=w_j(C_j)/Z_j$.  To obtain this stationary distribution, we use the Metropolis-Hasting technique \citep{Metropolis,10.1093/biomet/57.1.97}, with state space $\Omega_j:=\C{j}^{[\tilde{C}(v_1),...,\tilde{C}(v_j)]}$ and the Markov chain transitions $\tau: \Omega_j  \times\Omega_j \rightarrow [0,1]$ defined as follows. Let $C_j$ be the current state.
\begin{itemize}
    \item Choose a uniformly random $z\in\{j+1,\dots,n\}$. Let $C'_j$ be identical $C_j$, except $C'_j(v_z)=\bar{C_j}(v_z)$ (the assignment of $v_z$ is opposite in $C_j$ and $C'_j$).
    \item With probability $\min\{1,\frac{w_j(C'_j)}{w_j(C_j)}\}$ move to state $C'_j$. Otherwise, stay at $C_j$.
\end{itemize}
For more details, see Algorithm \ref{alg:sample_step}.
Therefore, $\tau(C_j,C'_j)=\frac{1}{(n-j)}\min\{1,\frac{w_j(C'_j)}{w_j(C_j)}\}$ and $\tau(C_j^{(1)},C_j^{(2)})=0$ for all other $C_j^{(1)},C_j^{(2)}$ where $C_j^{(1)} \neq C_j^{(2)}, C_j^{(1)}\neq C_j^{'(2)}$ (the self-loops $\tau(C_j,C_j)$ correspond to the remaining probability, so that $\tau$ is a stochastic matrix). 

The underlying transition graph is analogous to the hypercube with $n-j$ dimensions, and, therefore, the state space is connected (\ie, we can get from every state to every other state using transitions of the Markov chain). This Markov chain is also aperiodic due to the presence of self-loop transitions and, therefore, it has a unique stationary distribution. The Metropolis-Hastings technique ensures that this stationary distribution is exactly the distribution $\mu_j$, \ie, proportional to the weights $w_j(C_j)$.



Following the earlier outline (see (\ref{eq:self_red}) and (\ref{eq:partition})), we will be estimating $\sigma_j=\frac{Z_{j}}{Z_{j-1}}$. Recall that, for numerical stability, we wanted $\sigma_j\geq 1/2$. We showed that $\tilde{C}(v_j)=\arg\max_{s\in\{A,R\}}\sigma_j^s$ yields $\sigma_j = \max\{\sigma_j^A, \sigma_j^R\}\geq 1/2$. We will estimate both $\sigma_j^A$ and $\sigma_j^R$ simultaneously by generating samples from $\Omega_{j-1}=\C{j}^{[\tilde{C}(v_1),...,\tilde{C}(v_{j-1})]}$. A sample $C_{j-1}$ with $C_{j-1}(v_{j-1})=X$ will contribute to the $\sigma_j^X$ computation, for $X\in\{A,R\}$.

Let $C_{j-1}$ be a random sample from $\Omega_{j-1}$, drawn according to $\mu_{j-1}$. Define $f_A:\Omega_{j-1}\to\{0,1\}$ as the indicator function that the corresponding vertex assignment assigned $v_j$ as an activator: $f_A(C)=1$ if and only if $C(v_j)=A$.
Then,
\begin{eqnarray*}
&& \mathbf{E}_{\mu_{j-1}}[f_A(C_{j-1})] \\
&&= \sum_{C\in\Omega_{j-1}}\mu_{j-1}(C)f_A(C)
= \sum_{C\in\Omega_{j-1}:C(v_j)=A}\frac{w_{j-1}(C)}{Z_{j-1}}\\
&&= \frac{1}{Z_{j-1}}\sum_{C\in\C{j}^{[\tilde{C}(v_1),...,\tilde{C}(v_{j-1}),A]}}P(\A{}|\Theta, C) \cdot P(C|j-1) \\
&&= \frac{1}{Z_{j-1}}\sum_{C\in\C{j}^{[\tilde{C}(v_1),...,\tilde{C}(v_{j-1}),A]}}P(\A{}|\Theta, C) \cdot P(C|j)\q{} \\
&&= \frac{\q{}Z_j^A}{Z_{j-1}} = \q{}\sigma_j^A.
\end{eqnarray*}
Therefore, we can use the expectation of $f_A$ to estimate $\sigma_j^A$ (and, similarly, we can define $f_R$ and estimate $\sigma_j^R$ since $\mathbf{E}_{\mu_{j-1}}[f_R(C_{j-1})]=(1-\q{})\sigma_j^R$). The accuracy of the estimate increases with the number of samples: we will draw $k$ independent samples $C_j^{(1)},C_j^{(2)},\dots,C_j^{(k)}$ and compute the average $f_A(C_j)$, namely $\frac{1}{k}\sum_{i=1}^k f_A(C_j^{(i)})$. We use the same set of samples for the average $f_A$ and the average $f_R$ computation (each sample contributes to either $f_A$ or $f_R$). Since the expectation $\mathbf{E}_{\mu_{j-1}}[f_A(C_{j-1})]$ corresponds to the probability of drawing a sample with $v_j$ assigned as an activator, we draw such sample with probability $\frac{\q{}Z_j^A}{Z_{j-1}}$, and we draw a sample with $v_j$ as a repressor with probability $\frac{(1-\q{})Z_j^R}{Z_{j-1}}$. Since these two types of samples cover the state space, at least one of the types will be drawn with probability $\geq 1/2$. This means that we will be very likely to have a sufficient number of samples of at least one of the types, which will allow us to estimate the corresponding expected value closely. Since we do not know which of the types is more likely, we generate $k$ samples and estimate both the average $f_A$ and the average $f_R$. Then, we choose the larger average and define the corresponding $A$ or $R$ as $\tilde{C}(v_j)$. Then, divide the larger average by $\q{}$ or $(1-\q{})$ as appropriate, obtaining an estimate on $\sigma_j$.

This $\sigma_j$ might potentially be different than $\max\{\sigma_j^A, \sigma_j^R\}$ but we still have $\sigma_j\geq 1/2$:
Suppose $\q{}\sigma_j^A \geq (1-\q{})\sigma_j^R$. Then, $\q{}\sigma_j^A\geq \frac{1}{2}$ and we have a close estimate of this quantity. Then, divide this estimate by $\q{}$, obtaining an estimate of $\sigma_j$ where $\tilde{C}(v_j)=A$. We have $\sigma_j=\sigma_j^A\geq \frac{1}{2\q{}}\geq\frac{1}{2}$. Analogously, if $\q{}\sigma_j^A < (1-\q{})\sigma_j^R$, we obtain $\sigma_j=\sigma_j^R\geq \frac{1}{2(1-\q{})}\geq\frac{1}{2}$. 

Next we discuss the accuracy of this estimate. Suppose we aim to be within a $(1+\varepsilon)$-factor of the true $\Lh{\Theta}$, for some small $\varepsilon$. Since we have $n$ self-reducibility steps (\ie, $n$ quantities (ratios $\sigma_j$) to estimate, see (\ref{eq:self_red}) and (\ref{eq:partition})) and the estimated quantities are larger than $1/2$, we can follow the derivation in \citep{Jerrum}[Chapter 3, page 26] almost verbatim to obtain the desired number of samples per estimate. However, this number of samples is rather high for the datasets of our size, so we employ a convergence heuristics for the estimates, parameterized by an error parameter $\delta$, to reduce the empirical running time. The algorithm is summarized in Algorithm \ref{alg:BNC_lh}, where the last for-loop is explained in the following section.




\begin{figure}
\begin{algorithm}[H]
\footnotesize
\caption{ESTIMATE-BNC-LIKELIHOOD($G=(V,\A{})$, $\Theta$, $K$, $T$, $b$, $\delta$)}
\KwIn{A signed graph $G=(V,\A{})$, model parameters $\Theta$, number of samples $K$, mixing time $T$, convergence ``window'' $b$, convergence error $\delta$}
\label{alg:BNC_lh}
\begin{algorithmic}[1]
\STATE{Set $V_{\mathrm{done}} = \emptyset$}
\STATE{Throughout the algorithm, let $\tilde{C}: V_{\mathrm{done}} \rightarrow \{A,R\}$}
\STATE{Sort vertices of $V$ by their total degrees in descending order}
\FOR{$j=1$ to $n$}
\IF{$v_j \in V_{\mathrm{done}}$}
\STATE{BREAK}\tcp*{This vertex has already been assigned a group using analytical calculation}
\ENDIF
\STATE{Let $\mathrm{sum}_A = 0$ and $\mathrm{sum}_R = 0$.}
\FOR{$k=1$ to $K$}
  \STATE{C = \textit{$\mathrm{SAMPLE}(\A{},\Theta,\tilde{C},T)$}}
  \IF{$C(v_j)=A$}
  \STATE{Increment $\mathrm{sum}_A$}
  \ELSE
  \STATE{Increment $\mathrm{sum}_R$}
  \ENDIF
  \STATE{Let $S_k^A = \mathrm{sum}_A/k$ and $S_k^R = \mathrm{sum}_R/k$}
\IF{$|S_i^A-S_k^A|\leq\delta$ for every $i\in\{k-b,\dots,k-1\}$ or $|S_i^R-S_k^R|\leq\delta$ for every $i\in\{k-b,\dots,k-1\}$}
\STATE{BREAK}\tcp*{New samples are not changing the likelihood ratio}
\ENDIF
\ENDFOR
\STATE{$V_{\mathrm{done}} = V_{\mathrm{done}} \cup v_j$}
\IF{$S_k^A \geq \frac{1}{2}$}
\STATE{Let $\tilde{C}(v_j) = A$ and $S_j=S_k^A/\q{}$}
\ELSE
\STATE{Let $\tilde{C}(v_j) = R$ and $S_j=S_k^R/(1-\q{})$}
\ENDIF
\FOR{every neighbor $v_u$ of $v_j$} \IF{$v_u\not\in V_{\mathrm{done}}$ and all its neighbors are in $V_{\mathrm{done}}$}
\STATE{$V_{\mathrm{done}} = V_{\mathrm{done}} \cup v_u$}
\STATE{Use (\ref{eq:analytical}) to calculate $S_u$ analytically}
\IF{$S_u \geq \frac{1}{2}$}
    \STATE{Let $\tilde{C}(v_u)=A$ and $S_u=S_u/\q{}$}
\ELSE
    \STATE{Let $\tilde{C}(v_u)=R$ and $S_u=S_u/(1-\q{})$}
\ENDIF
\ENDIF
\ENDFOR
\ENDFOR
\STATE{Use (\ref{eq:lh_one_c}) to compute $Z_n = \Lh{\Theta|\tilde{C}}$}
\RETURN{$\frac{Zn}{\prod_{j=1}^n S_j}$}
\end{algorithmic}
\end{algorithm}
\caption{MCMC likelihood estimation of the BNC model for parameters $\Theta$. Sorting the vertices leads to faster convergence of the estimates.}
\end{figure}


\begin{figure}
\begin{algorithm}[H]
\footnotesize
\caption{SAMPLE($\A{}$, $\Theta$, $C_p$, $T$)}
\KwIn{An adjacency relation $\A{}$, model parameters $\Theta$, a partial vertex assignment $C_p:\{v_1,\dots,v_{j-1}\} \rightarrow \{A,R\}$, 
and mixing time $T$}
\label{alg:sample_step}
\begin{algorithmic}[1]
\STATE{Let $C_j$ be $C_p$ with a uniformly random assignment on vertices $v_j,\dots,v_n$}\tcp*{Choose an initial state}
\FOR{$t=1$ to $T$}
  \STATE{Choose  $z \in \{j, \dots, n\}$ uniformly at random}
 \STATE{Let $C'_j$ be identical to $C_j$, except $C'_j(v_z)=\bar{C}_j(v_z)$}
\STATE{Let $\mathrm{ratio} := \frac{1-q_A}{q_A} \cdot \prod_{X \in \{A,R\}} (\frac{\xi_{R,X}}{\xi_{A,X}})^{\hcoutc{+}{X}{v_{z}}} \cdot (\frac{1-\xi_{R,X}}{1-\xi_{A,X}})^{\hcoutc{-}{X}{v_{z}}} \cdot (\frac{\xi_{X,R}}{\xi_{X,A}})^{\hcinc{+}{X}{v_{z}}} \cdot (\frac{1 - \xi_{X,R}}{1-\xi_{X,A}})^{\hcinc{-}{X}{v_{z}}}\cdot \frac{\xi_{R,R}}{\xi_{A,A}}^{\mathrm{self}_{+}(v_z)}\cdot \frac{1-\xi_{R,R}}{1-\xi_{A,A}}^{\mathrm{self}_{-}(v_z)}$}
  \IF{$C(v_z)=A$ in $C_j$}
  \STATE{$\tau(C_j, C'_j) = \min(1,\frac{w_j(C'_j)}{w_j(C_j)}) =  \min(1, \mathrm{ratio})$}
  
  \ELSE 
  \STATE{$\tau(C_j, C'_j) = \min(1,\frac{w_j(C_j)}{w_j(C'_j)}) =  \min(1, \frac{1}{\mathrm{ratio}})$}
\ENDIF  
    \STATE{With probability $\tau(C_j, C'_j)$ set $C_j = C'_j$}
\ENDFOR
\RETURN{$C_j$}
\end{algorithmic}
\end{algorithm}
\caption{A Markov chain to produce a sample vertex assignment from $C_j\in \C{j}^{[\tilde{C}(v_1),...,\tilde{C}(v_j)]}$, generated with probability proportional to $w_j(C_j)$. $\hcoutc{s}{X}{v_z}$  denotes the number of outgoing edges of sign $s$ from $v_z$ to vertices with assignment $X\in\{A,R\}$. $\hcinc{s}{X}{v_z}$ denotes the number of incoming edges of sign $s$ from vertices assigned $X$ to $v_z$. Neither $\hcoutc{s}{X}{v_z}$ and $\hcinc{s}{X}{v_z}$ include the self-loops.}
\end{figure}

\begin{figure}
\begin{algorithm}[H]
\footnotesize
\caption{calculate ratio analytically}
\KwIn{Vertex $v_u$, partial node partition $C_p: V_p \rightarrow S$, where $V_p \subseteq V$, parameters $\Theta$}
\label{alg:analytic_calc_ratio}
\begin{algorithmic}[1]
\STATE{$V_p = V_p \cup v_u$}
\STATE{Let $\beta_u^a=1$}
\STATE{Let $\beta_u^r=1$}
\FOR{edge $e \in E(v_u)$} 
\STATE{Let $p_a$ be the probability of edge $e$ if $v_u$ is activator}
\STATE{Let $p_r$ be the probability of edge $e$ if $v_u$ is repressor}
\STATE{$\beta_u^a = \beta_u^a * p_a$}
\STATE{$\beta_u^r = \beta_u^r * p_r$}
\ENDFOR
\STATE{Let $\sigma_u^a = \frac{q\beta^a_u}{q\beta^a_u + (1-q)\beta^r_u}$}
\STATE{Let $\sigma^r_u = 1-\sigma^a_u$}
\IF{$\sigma^a_u > \sigma^r_u$}
\STATE{$C_p(v_u)=A$}
\STATE{$\sigma_j^* = \sigma^a_u$}
\ELSE{
\STATE{$C_p(v_u)=R$}
\STATE{$\sigma_j^* = \sigma^r_u$}
}
\ENDIF
\end{algorithmic}
\end{algorithm}
\caption{A function to analytically calculate the ratio for the vertex that is connected only to already assigned vertices}
\end{figure}

\subsection{Algorithmic speed-up of MCMC}
We combine the MCMC approach with exact calculation to achieve significant speed-up for real-world data sets. 
The GRNs, as many other real-world networks, are not dense: they usually have a few highly connected nodes (called \textit{hubs} in network science) and many low-degree nodes. For such low-degree nodes, the use of MCMC is unnecessary and computationally wasteful: as soon as all of their adjacent vertices are assigned, the ratio of likelihoods for such nodes can be calculated analytically. 

\begin{theorem}
\label{thm:speedup}
Let $C:V\to\{A,R,\mathrm{undefined}\}$ be a partial vertex assignment. Reorder the vertices so that $v_1,\dots,v_{j-1}$ are assigned by $C$ to either $A$ or $R$, and all other vertices are undefined. Suppose that all of $v_j$'s neighbors are already assigned by $C$. 
Then $\sigma_j^*$ can be calculated as:
\begin{equation}
    \sigma_j^* = \max\{\frac{\q{}\beta_j^A}{\q{}\beta_j^A + (1-\q{})\beta_j^R}, \frac{(1-\q{})\beta_j^R}{\q{}\beta_j^A + (1-\q{})\beta_j^R} \}, \label{eq:analytical}
\end{equation}
where $\beta_j^A(C)$ is the product of probabilities of all adjacent edges to vertex $v_j$ (except self-loops), driven by vertex assignment $C$ and $v_j$ being an activator. 

Formally, $\beta_j^A(C)=\prod_{i<j:(v_i,v_j)\in E}\pr{\A{}(v_i,v_j)}_{C(v_i),A)}\prod_{i<j:(v_j,v_i)\in E}\pr{\A{}(v_j,v_i)}_{C(A,v_i))}$.
Similarly, $\beta_j^R(C)$ is a product of probabilities of $v_j$'s edges when $v_j$ is a repressor.  
\end{theorem}
The proof appears in the Appendix. The theorem allows us to compute $\sigma_j^*$ exactly  in time proportional to $v_j$'s degree.

For a mixing time $T$ and a number of samples $K$, the running time of Algorithm \ref{alg:BNC_lh} is $\Oh{n^2KT}$: Each step of the Markov chain can be implemented in $\Oh{n}$ time and, thefore, SAMPLE (see Algorithm \ref{alg:sample_step}) takes time $\Oh{Tn}$. In fact, the $\Oh{n}$-term can be reduced to $\Oh{\deg(v_z)}$ time, which for typical real-world data amortizes to $\Oh{1}$.
Adding in the loops in Algorithm \ref{alg:BNC_lh} of $n$ and $K$ iterations, respectively, we get the stated time bound. Moreover, applying Theorem \ref{thm:speedup}, for real-world data, for the vast majority of vertices we can reduce the $\Oh{Tn}$ term to $\Oh{\mathrm{\deg}}$, which further amortizes to $\Oh{1}$, vastly improving the running time. We note that for real-world data, even though we computed the desired $K$ theoretically following the calculation in \citep{Jerrum}, we found that we did not need to generate all $K$ samples since the convergence heuristics involving $b$ and $\delta$ terminated the loop significantly earlier. Moreover, these heuristical estimates were very close to the target values for all manually checked data (including large data sets where we knew the exact values, for example, when the parameters fall under the NO, SC, or TC models), validating the choice of our convergence heuristics parameters $b$ and $\delta$ and our mixing time estimate $T$ (also obtained heuristically).






\begin{table}[t!]
    \footnotesize
    \centering
    \begin{tabular}{ccccc}
    Type & \# of nodes & \# of edges & \# of $+$ edges & \# of $-$ edges \\
    \toprule
    Regulon (\textit{E coli}) & 1922 & 4265 & 2256 & 2000 \\
    
    SubtiWiki (\textit{Bacillus subtilis}) & 2563 & 5283 & 3436 & 1847  \\
    
    Synthetic (sampled from DSF) &  2000 & Median*: 3100, CI: (2938; 3275) & NA** & NA** \\
    
    \end{tabular}
    \caption{Network characteristics of real-world and synthetic datasets. 9 edges of unknown types were filtered out from Regulon database  * The median and confidence interval over 40 generated topologies from DSF 
    ** The number of signed edges varied substantially based on $\Theta$-generator}
    \label{tab:network_descr}
\end{table}

\section{Empirical Evaluation}
\subsection{Data Collection and setup}
\paragraph{Synthetic datasets}
To validate that our algorithms can correctly identify the GRASMOS parameters, we created synthetic datasets with known $\Theta$, resembling real-world GRNs by their network parameters and size. The characteristics of all synthetic and real-world datasets we used are shown in Table \ref{tab:network_descr}. Our evaluation pipeline had two steps: generation and reconstruction. During the generation part, we generated network topologies from the Directed Scale Free (DSF) model, known to well capture the characteristics of GRNs \citep{van2006syntren}, then using several different $\Theta$ parameters we sampled the vertex group partition and assigned $+$ and $-$ signs to the edges. The DSF model is an iterative generator that grows the networks by adding the vertices and edges until the desired size \citep{bollobas2003directed}. At each step it performs one of these three steps: (1) with probability $\alpha$ it adds a new vertex with an edge to an existing vertex $u$, where $u$ is chosen proportionally to its in-degree and an \textit{in-degree bias term $\delta_{in}$}; (2) with probability $\gamma$ it adds a new vertex with an edge from an existing vertex $u$, where $u$ is chosen proportionally to its out-degree and an \textit{out-degree bias term $\delta_{out}$}; (3) with probability $\beta$ it adds an edge between a pair of existing vertices, where a source node is chosen according to its out-degree and $\delta_{out}$, and the target node according to its in-degree and $\delta_{in}$. For our experiments, we set $\alpha = 0.41$, $\beta = 0.49$, $\gamma = 0.1$,  $\delta_{in} = 0$, $\delta_{out} = 0.05$ following the works of \citep{bollobas2003directed} and \citep{van2006syntren}. We filtered out any duplicated edges from the graph instances to meet our model specification. 

For illustrative purposes, in the Results section below we present our validation for the source-consistent model: We chose several $\Theta$ (we refer to them as $\Theta$-generators) from this model and generated corresponding signages for our graph. 
During the reconstruction part, we tested multiple candidate $\Theta$ ($\Theta$-candidates) belonging to the SC, TC, and NO models, checked the proportion of samples in which the candidate with the best likelihood $\Theta_{\mathrm{MLE}}$ matched the $\Theta$-generator, and calculated the $L_1$-norm between the $\Theta$-generator and $\Theta_{\mathrm{MLE}}$. The parameters within the $\Theta$-candidates varied between $0.1$ and $0.9$, with increment step of $0.1$. 
We tested $4$ different $\Theta$-generators in the generation part, and for each of them 1458 of $\Theta$-candidates in the reconstruction part. To account for stochasticity we repeated the pipeline $10$ times.

\paragraph{Real-world GRNs}
We used information on gene regulation from two public databases: RegulonDB \citep{santos2019regulondb} and SubtiWiki \citep{pedreira2022current, florez2009community}. Both databases contain experimentally validated information about gene regulatory interactions and their type (activation/repression) for a single bacteria species: Regulon for \textit{Escherichia coli} and SubtiWiki for \textit{Bacillus subtilis}. We only left the entries associated with transcriptional gene regulation. From both datasets we filtered our regulation edges of unknown type and duplicated edges. Regulation interactions can be context-dependent: under different conditions the same gene can either activate or repress the target gene \citep{ong2011enhancer}. Our model do not account for such scenarios, so we had to filter out such duplicated edges.


For fitting the GRASMOS parameters of real-world GRNs, we: used a fine-grained exploration of parameters varying in $[0.1,0.9]$ with increment of $0.05$ belonging to the SC, TC, and NO models; and, for complexity reasons, we evaluated the BNC model subspace using a coarse-grid with each parameter varying in $\{0.25, 0.5, 0.75\}$. For $\Theta$s belonging to multiple models, we compared the results from these approaches, made sure that they were consistent, and identified a good candidate to run a ``refined'' BNC search with $432$ additional $\Theta$-candidates in the vicinity (searching through $\Theta$s roughly within $\pm 0.125$ from the identified candidate).

We estimated $\Lh{\Theta}$ of the BNC parameters using a parallel implementation on our university's computing cluster. 
We used up to $432$ nodes with each core estimating a likelihood of a single $\Theta$-candidate via MCMC sampling (Algorithm \ref{alg:BNC_lh}). Each core is equipped with Intel\textsuperscript\textregistered Xeon\textsuperscript\textregistered Gold 6150 CPU @ 2.70GHz. The RAM upper limit for our computation was 2048 MB.

\section{Results and Discussion}
\subsection{Synthetic data}
For synthetic datasets, we tested how well our framework based on the total likelihood reconstructs the $\Theta$ values that were used to generate an edge signage instance. We found that, overall, it was able to reconstruct the  $\Theta$-generator 
well. Figure \ref{fig:theta_recovery} shows the percentage of test cases when the $\Theta$-generator equaled the reconstructed $\Theta_{\textrm{MLE}}$ exactly, and when the $\xi$-portion of $\Theta$ and $\Theta_{\textrm{MLE}}$ were identical, respectively. Reconstruction of the $\xi$-parameters was particularly successful. Moreover, even when $\Theta$ and $\Theta_{\textrm{MLE}}$ differed, they were very close in their $L_1$-norm, see Figure \ref{fig:norm_recovery}.


\begin{figure}
\begin{minipage}{.45\textwidth}
    \centering
    \includegraphics[width=\columnwidth]{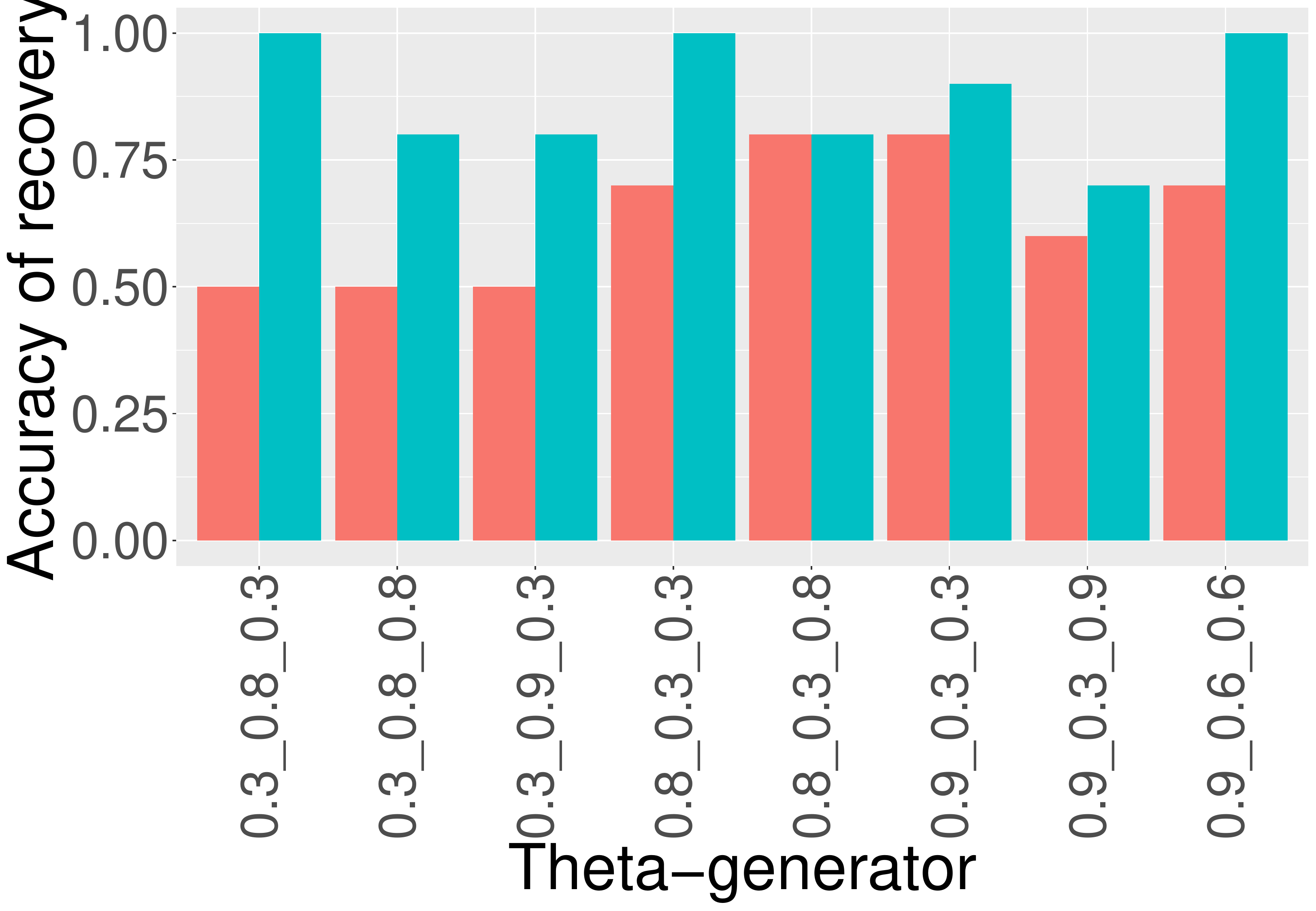}
    \caption{The absolute accuracy of $\Theta$ (red) and $\E$ (blue) recovery for SC model. The format of the labels is $\xi_{A,*}$ $\xi_{R,*}$, $\q{}$}
    \label{fig:theta_recovery}
    \end{minipage}
    \hspace{0.5cm}
    \begin{minipage}{.45\textwidth}
    \centering
    \includegraphics[width=\columnwidth]{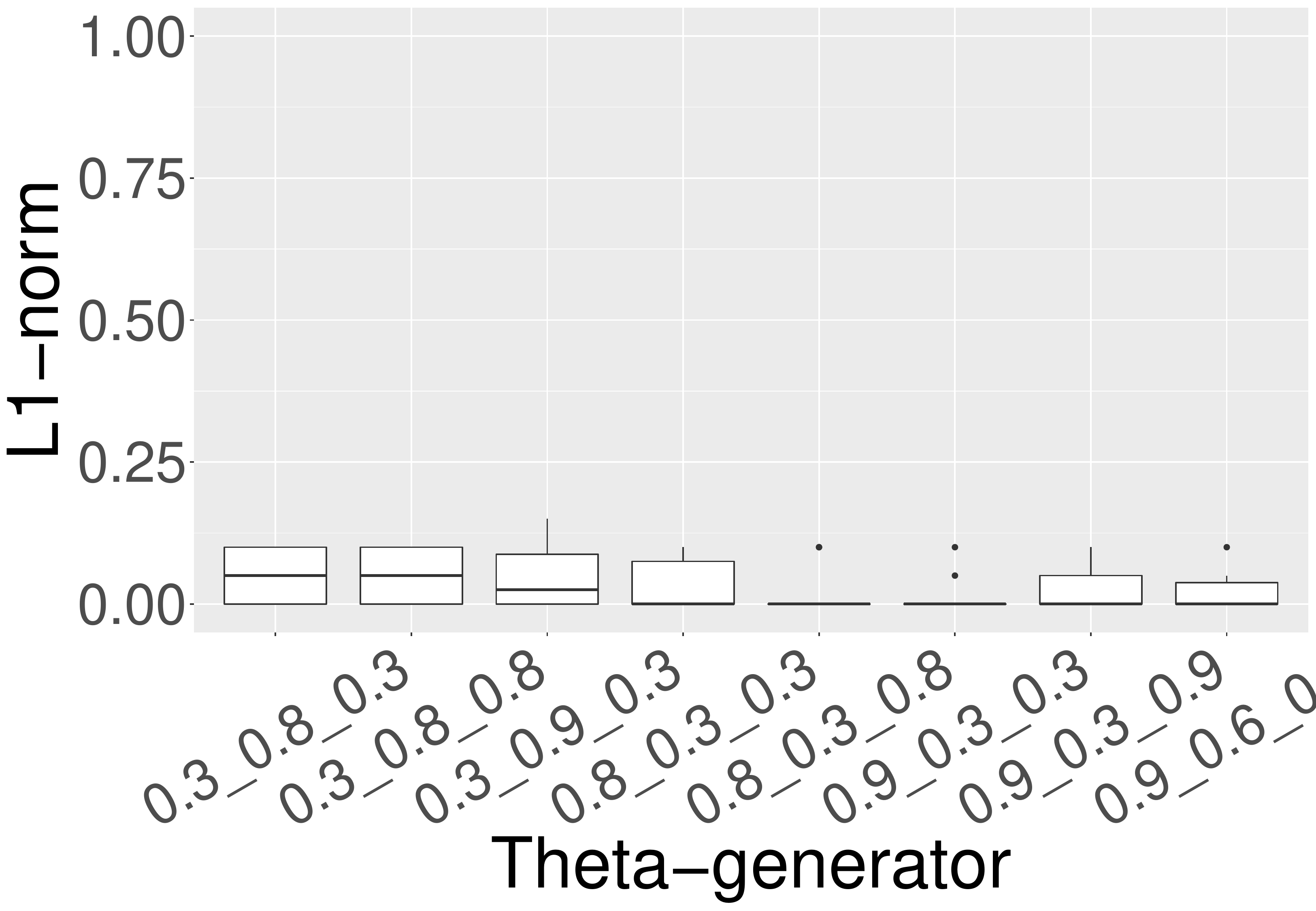}
    \caption{The accuracy of the $\Theta$ recovery in SC: The average and the standard deviation
    of the $L1$-norm between the $\Theta$-generator and the best among the tested $\Theta$-candidates. The format of the labels is $\xi_{A,*}$ $\xi_{R,*}$, $\q{}$}
    \label{fig:norm_recovery}
    \end{minipage}
\end{figure}



\subsection{Results on real-world GRNs}
Table \ref{tab:best_lh_bothdf} presents 5 best $\Theta$ candidates per model for the SubtiWiki and Regulon datasets. For both datasets, the BNC parameters had the best performance among all models. However, some of the top BNC parameters actually belonged to the SC model for SubtiWiki. The source-consistency hypothesis corresponds well to the existence of \textit{operons} -- groups/clusters of bacterial genes that are co-located on the DNA strand and controlled by the same gene-regulator \citep{salgado2001regulondb}. Interestingly, in case of Regulon, none of the top-5 $\Theta$ belonged to the SC model. Moreover, all of the top parameters implied that genes tend to activate the genes from groups other than they belong to. This observation might corroborate the existence of feed-forward transcriptional control with interchanging edge signs suggested for certain GRNs \citep{sasse2015feed} and raises the question about whether the bacteria species corresponding to these two data sets indeed have different gene regulation mechanisms, or the discrepancies should be explained by other factors (missing edges, noise associated with data curation, \etc).

\begin{table}[!t]
\centering
\begin{tabular}{rlrrrrrr}
  \toprule
  \hline
 & model & $\xi_{AA}$ & $\xi_{AR}$ & $\xi_{RA}$ & $\xi_{RR}$ & $\q{}$ & $\Lh{\Theta}$ \\ 
  \toprule
1 & BNC & 0.99 & 0.99 & 0.20 & 0.15 & 0.40 & 557.12 \\ 
  2 & BNC/SC & 0.99 & 0.99 & 0.15 & 0.15 & 0.40 & 557.71 \\ 
  3 & BNC & 0.99 & 0.99 & 0.20 & 0.15 & 0.50 & 558.71 \\ 
  4 & BNC/SC & 0.99 & 0.99 & 0.15 & 0.15 & 0.50 & 558.93 \\ 
  5 & BNC & 0.99 & 0.99 & 0.15 & 0.20 & 0.40 & 559.24 \\ 
  \midrule
  6 & SC & 0.95 & 0.95 & 0.10 & 0.10 & 0.45 & 581.71 \\ 
  7 & SC & 0.95 & 0.95 & 0.10 & 0.10 & 0.50 & 581.88 \\ 
  8 & SC & 0.95 & 0.95 & 0.10 & 0.10 & 0.40 & 582.42 \\ 
  9 & SC & 0.95 & 0.95 & 0.10 & 0.10 & 0.55 & 582.91 \\ 
  10 & SC & 0.95 & 0.95 & 0.15 & 0.15 & 0.45 & 583.49 \\ 
  \midrule
  11 & NO & 0.65 & 0.65 & 0.65 & 0.65 & NA & 1484.93 \\ 
  \midrule
  12 & TC & 0.65 & 0.70 & 0.65 & 0.70 & 0.95 & 1485.03 \\ 
  13 & TC & 0.70 & 0.65 & 0.70 & 0.65 & 0.05 & 1485.03 \\ 
  14 & TC & 0.60 & 0.65 & 0.60 & 0.65 & 0.05 & 1485.07 \\ 
  15 & TC & 0.65 & 0.60 & 0.65 & 0.60 & 0.95 & 1485.07 \\ 
  16 & TC & 0.65 & 0.70 & 0.65 & 0.70 & 0.90 & 1485.19 \\ 
  \hline
  \\
  \toprule
  \hline
  & model & $\xi_{AA}$ & $\xi_{AR}$ & $\xi_{RA}$ & $\xi_{RR}$ & $\q{}$ & $\Lh{\Theta}$ \\ 
  \toprule
1 & BNC & 0.70 & 0.80 & 0.20 & 0.15 & 0.50 & 1127.74 \\ 
  2 & BNC & 0.70 & 0.80 & 0.25 & 0.10 & 0.50 & 1128.75 \\ 
  3 & BNC & 0.70 & 0.80 & 0.20 & 0.10 & 0.50 & 1129.03 \\ 
  4 & BNC & 0.70 & 0.80 & 0.25 & 0.15 & 0.50 & 1129.74 \\ 
  5 & BNC & 0.75 & 0.75 & 0.20 & 0.15 & 0.50 & 1129.84 \\ 
  \midrule
  6 & SC & 0.75 & 0.75 & 0.15 & 0.15 & 0.50 & 1130.33 \\ 
  7 & SC & 0.75 & 0.75 & 0.20 & 0.20 & 0.50 & 1130.45 \\ 
  8 & SC & 0.75 & 0.75 & 0.15 & 0.15 & 0.55 & 1130.51 \\ 
  9 & SC & 0.75 & 0.75 & 0.20 & 0.20 & 0.45 & 1130.78 \\ 
  10 & SC & 0.75 & 0.75 & 0.20 & 0.20 & 0.55 & 1130.91 \\ 
  \midrule
  11 & NO & 0.53 & 0.53 & 0.53 & 0.53 & NA & 1277.51 \\ 
  \midrule
  12 & TC & 0.55 & 0.50 & 0.55 & 0.50 & 0.60 & 1412.39 \\ 
  13 & TC & 0.50 & 0.55 & 0.50 & 0.55 & 0.40 & 1412.39 \\ 
  14 & TC & 0.55 & 0.50 & 0.55 & 0.50 & 0.55 & 1412.40 \\ 
  15 & TC & 0.50 & 0.55 & 0.50 & 0.55 & 0.45 & 1412.40 \\ 
  16 & TC & 0.55 & 0.50 & 0.55 & 0.50 & 0.65 & 1412.44 \\ 
  \hline
\end{tabular}
\caption{Five best $\Theta$s per model for SubtiWiki (top) and Regulon (bottom) dataset. The lower $\Lh{\Theta}$, the better.}
\label{tab:best_lh_bothdf}
\end{table}

Also, we assessed the accuracy of the MCMC sampling for the likelihood estimation of those $\Theta$ instances for which we have the exact solution (\ie, instances of the SC, TC, and NO models). We aimed to have a $(1+ 1/n)$-multiplicative accuracy, \ie, at most $0.2\%$ likelihood error for our data sets. (We used this target accuracy for our computation of the number of needed samples, $K$.) In most cases our MCMC computations were within the target accuracy.
In those cases when we were off by more than $(1+ 1/n)$, the corresponding $\Theta$-candidates had $\Lh{\Theta}$ far from $\Lh{\Theta_{\mathrm{MLE}}}$, and the likelihoods for the $\Theta$s on the BNC coarse grid were vastly different. This meant that despite having more inaccurate $\Lh{\Theta}$ estimates for these (few) $\Theta$s than we hoped for, the coarse search eliminated these $\Theta$s from consideration and therefore eliminated the inaccuracies. In all cases when the $\Theta$-candidates were close to $\Theta_{MLE}$, the MCMC sampling accuracy was within the target accuracy (Figures \ref{fig:MCMC_accuracy_SubtiWiki}, \ref{fig:MCMC_accuracy_regulon}).

 \begin{figure}
    \centering  \includegraphics[height=8.5cm]{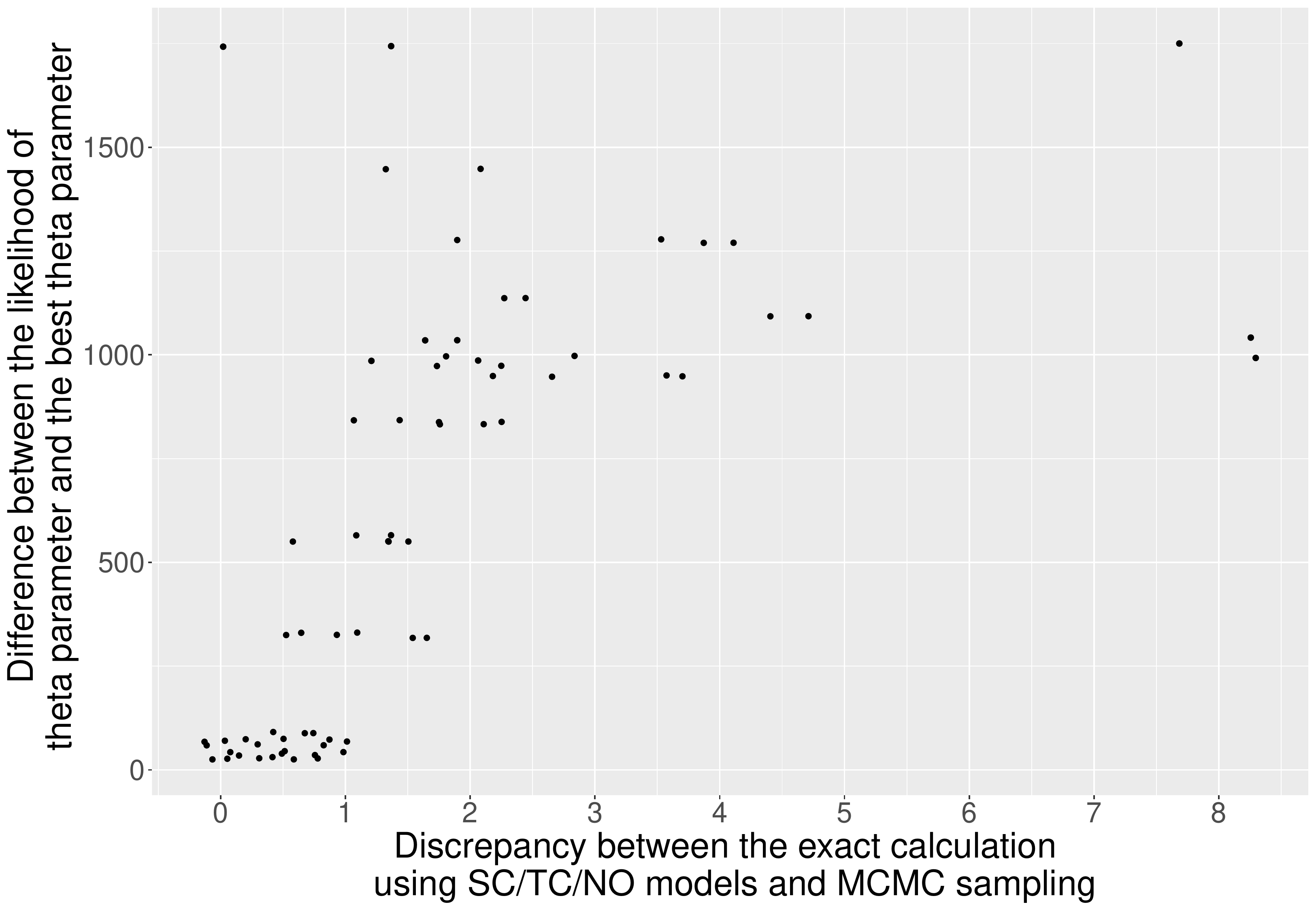}
    \caption{The accuracy of the MCMC estimation of $\Lh{\Theta}$ for SubtiWiki. Each point corresponds to a single instance of $\Theta$, where $\Lh{\Theta}$ can be computed exactly. The $x$-axis of the point corresponds to $|\Lh{\Theta}_{\mathrm{MCMC}}-\Lh{\Theta}_{\mathrm{exact}}|$ (accuracy of the estimate), and the $y$-axis to $|\Lh{\Theta}_{\mathrm{MCMC}}-\Lh{\Theta_{\mathrm{MLE}}}|$ (proximity to the optimal likelihood). The $\Theta$-candidate points with large $x$-coordinates also have a large $y$-coordinate, \ie, are far from the optimal likelihood, and get eliminated in the refined search. All $\Theta$-candidates close to $\Lh{\Theta_{MLE}}$ have accuracy within $1$ of the true value.}
    \label{fig:MCMC_accuracy_SubtiWiki}
    \end{figure}
    
    \begin{figure}
     \centering
     \includegraphics[height=8.5cm]{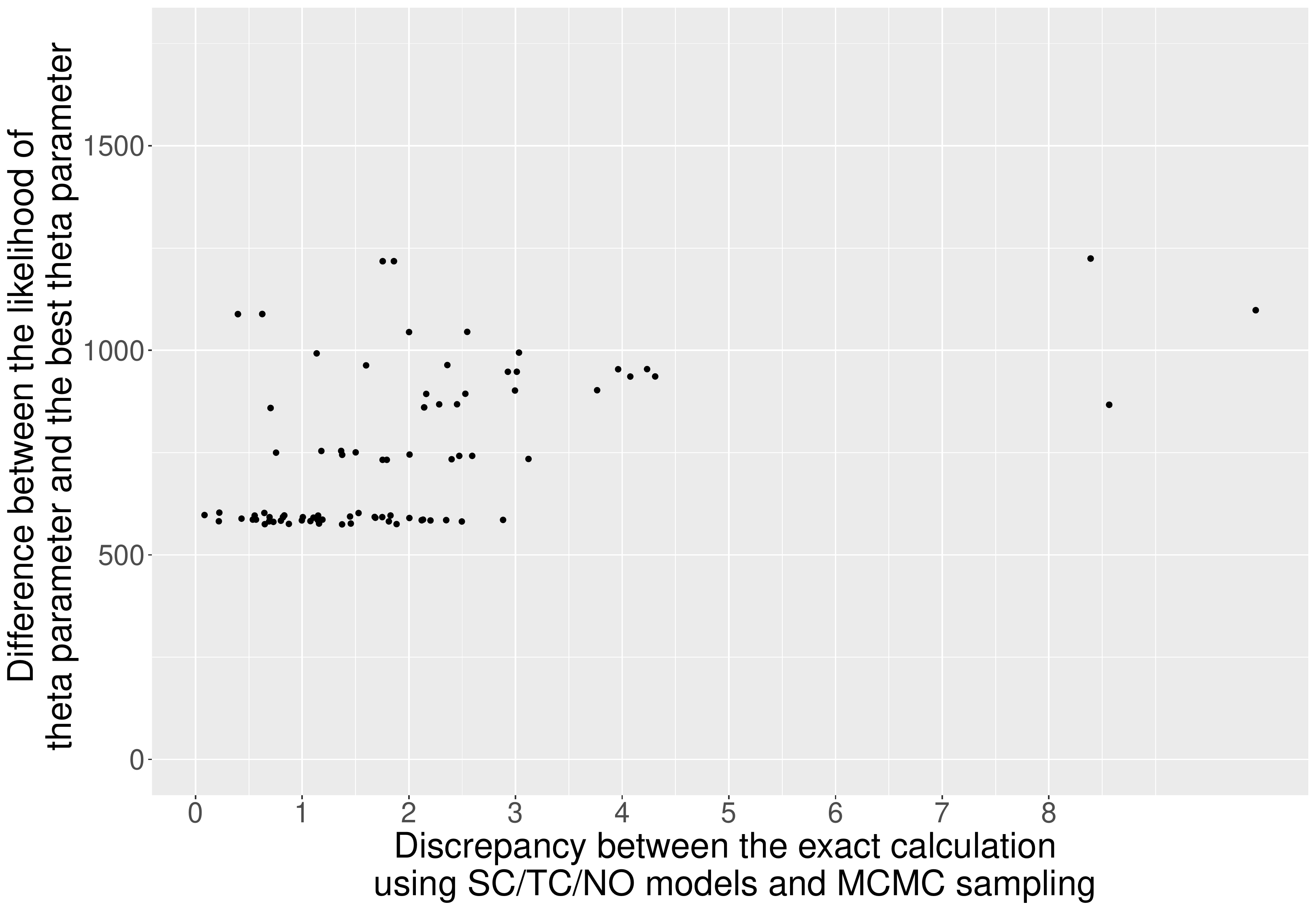}
    \caption{The accuracy of the MCMC estimation of $\Lh{\Theta}$ for Regulon, see Figure \ref{fig:MCMC_accuracy_SubtiWiki} for the plot description. 
    Compared to the SubtiWiki dataset, we suspect that the accuracy is worse, because the $\Theta$-candidates are further from the $\Theta_{MLE}$.}
    \label{fig:MCMC_accuracy_regulon}
    \end{figure}

\subsection{Running time} The analytical calculation of some vertices (Section Algorithmic speed-up of MCMC) and the reduction of the parameter space by half (Section Problem space symmetry) resulted in 16x actual speed-up of the likelihood estimation of BNC $\Theta$ parameters compared to the naive approach. Due to parallelization of the grid search, the total runtime of fitting the BNC parameters was around 2 days for Regulon, and 4 days for the SubtiWiki. You can find the details of the running time for BNC, SC/TC models in Table \ref{tab:runtime}.

\begin{table}[H]
\centering
\footnotesize

\begin{tabular}{llrlrrl}
 Dataset & Run & \#$\Theta$ & \# cores & Total time,sec & Mean time per $\Theta$,sec & Confidence interval,sec \\ 
  \hline
Regulon & BNC coarse & 243 & 243 & 73782 & 37838.04 & (23324;65431) \\ 
Regulon & BNC refined & 432 & 432 & 30873 & 26439 & (22694;30064) \\ 
SubtiWiki & BNC coarse & 243 & 243 & 88416 & 44869 & (25956;77424) \\ 
SubtiWiki & BNC refined & 432 & 432 & 25511 & 22021 & (19912;24239) \\ 
SubtiWiki & SC/TC & 13718& 1 & 347.18 & 0.03 & (0.02;0.04) \\ 
Regulon & SC/TC & 13718 & 1 & 332.57 & 0.02 & (0.01;0.08) \\ 
  \hline
\end{tabular}
\caption{Summary of the running time for estimating $\Lh{\Theta}$ for BNC (coarse/refined grid search) and SC/TC models. BNC parameters were fit using one core per parameter, SC/TC uses a single core for all parameters.}
 \label{tab:runtime}
\end{table}

\section{Conclusion}
In this work we stated a novel Maximum-Likelihood-based the total model selection problem for gene regulation, GRASMOS, developed a fitting framework for the problem and showcased its usage for two gene regulatory networks of $B. subtilis$ and $E. coli$. Our graph signage models and the model selection framework opened up a plethora of directions for future research. Among them:

\begin{itemize}
    \item Proving the NP-hardness of the BNC $\Theta$ parameters estimation.
    \item We used the master vertex assignment $\tilde{C}$ (visually shown for SubtiWiki dataset in Figure \ref{fig:GRN_SubtiWiki}) to estimate the likelihood, but how is it related to the vertex assignment with the maximum likelihood? In other words, is there any connection of $\tilde{C}$ to $C_{MLE,\Theta} = \arg \max_{C \in \C{}}\Lh{\Theta|C}$?
    \item In the same $\tilde{C}$ of SubtiWiki, we noticed that for $\Theta_{MLE}$ there were many vertices $v_x$ for which the probability of distribution $\tilde{C}$ was uniform, i.e., they had exactly 50\% chance of being assigned an activator during the course of the BNC algorithm. We think it happens because such vertices have low degrees, and the model correctly identified that there is not enough information for assigning them to specific groups. Are there topology-aware models that take vertex degrees into account and achieve a better fit than our current (topology-agnostic) models?

    \item Is our framework robust for missing edges? What is the percentage of the missing edges it can handle, if the edges are removed uniformly at random? And if they are removed proportionally to their degrees?
    \item Does the high explanatory power of the $\Theta_{MLE}$ translate into high predictive power?
\end{itemize}

\begin{figure}[H]
    \centering
    \includegraphics[width=\columnwidth]{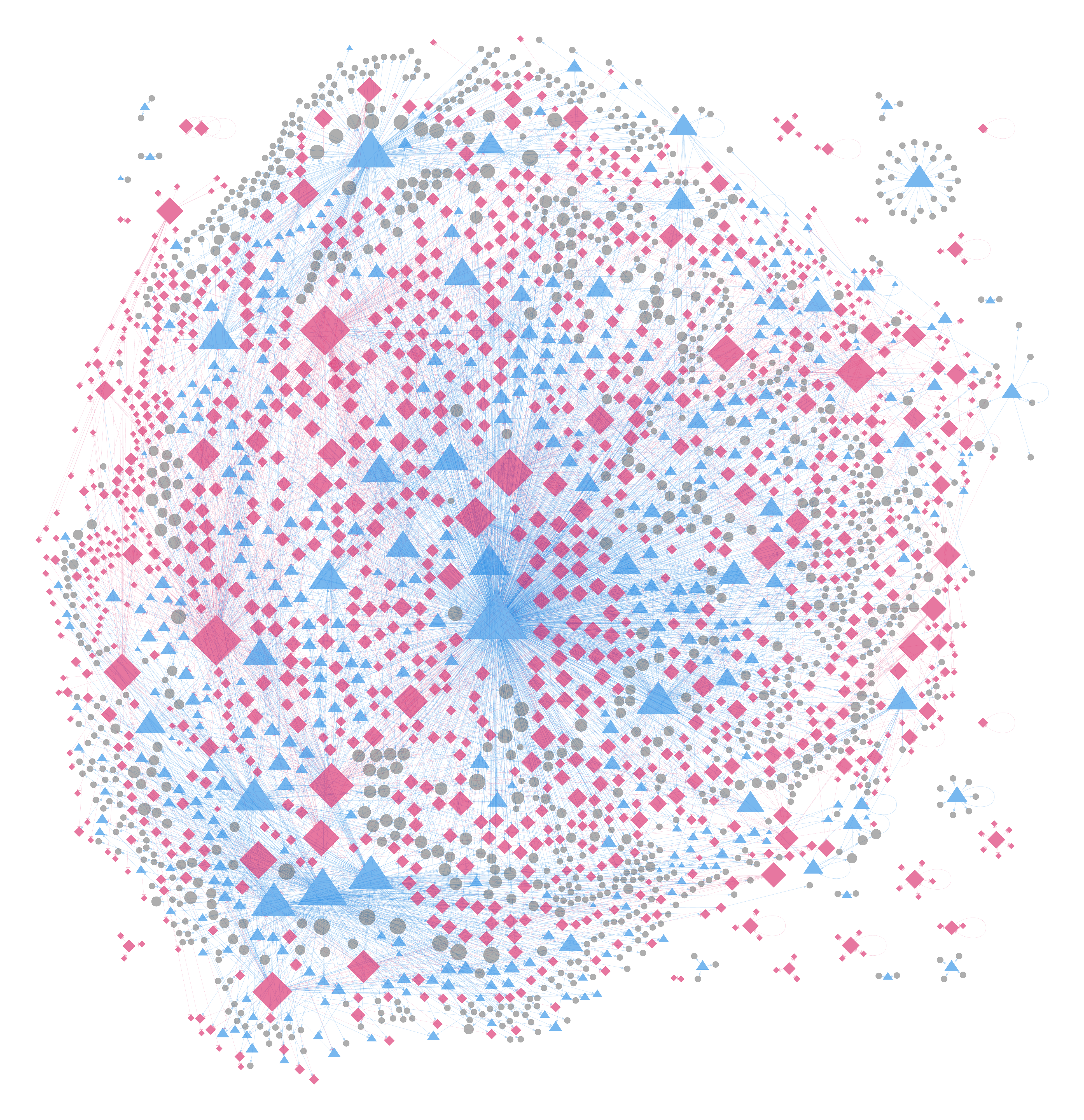}
    \caption{SubtiWiki signed GRN with the ``master'' vertex assignment $\tilde{C}$ corresponding to $\Theta_{MLE}$.
    Red diamonds - repressors in $\tilde{C}$, blue triangles = activators in $\tilde{C}$, gray circles = vertices where the assignment of $\tilde{C}$ is ambiguous, corresponding to probability exactly $1/2$ of being activator or repressor in the corresponding node assignment subspace. The symbol sizes are proportional to the total degree of the vertices.}
    \label{fig:GRN_SubtiWiki}
\end{figure}

\bibliographystyle{plainnat}
\bibliography{ref}
\end{document}